\documentclass[11pt]{article}

\usepackage[utf8]{inputenc}

\usepackage{amsmath,amsfonts,amsthm,amssymb,color}
\usepackage[usenames,dvipsnames,svgnames,table]{xcolor}
\definecolor{darkgreen}{rgb}{0.0,0,0.9}
\usepackage{mathtools}
\usepackage{authblk}
\usepackage{fullpage}
\usepackage{parskip}
\usepackage{comment}
\usepackage{tikz}
\usepackage{bbm}
\usepackage{dsfont}
\usepackage[sc]{mathpazo}
\usepackage[basic]{complexity}
\usepackage{algorithm2e}
\usepackage[colorlinks=true,
citecolor=OliveGreen,linkcolor=BrickRed,urlcolor=BrickRed,pdfstartview=FitH]{hyperref}
\usepackage[capitalize,nameinlink]{cleveref}

\SetKwInOut{Input}{Input}
\SetKwInOut{Output}{Output}
\SetKwFunction{Uncross}{\textsc{Uncross}}
\SetKwFunction{MergeUncross}{\textsc{MergeUncross}}
\SetKwFunction{ConnectedComponents}{\textsc{ConnectedComponents}}
\SetKwFunction{PerfectMatching}{\textsc{PerfectMatching}}
\SetKwBlock{InParallel}{in parallel do}{end}
\SetKwFor{ParallelFor}{for}{in parallel do}{end}

\newcommand*{\suppress}[1]{}

\makeatletter
\def\thm@space@setup{%
	\thm@preskip= 10pt
	\thm@postskip=\thm@preskip 
}
\makeatother

\makeatletter
\renewcommand{\paragraph}{%
	\@startsection{paragraph}{4}%
	{\z@}{5pt}{-1em}%
	{\normalfont\normalsize\bfseries}%
}
\makeatother

\newtheorem{theorem}{Theorem}
\newtheorem{lemma}{Lemma}

\newtheorem{proposition}[theorem]{Proposition}

\theoremstyle{definition}

\newenvironment{fminipage}%
{\begin{Sbox}\begin{minipage}}%
		{\end{minipage}\end{Sbox}\fbox{\TheSbox}}

\def\union{\cup}
\def\intersect{\cap}

\newcommand\ff{\boldsymbol{\mathit{f}}}
\renewcommand\gg{\boldsymbol{\mathit{g}}}

\newcommand\pp{\boldsymbol{\mathit{p}}}

\newcommand\yy{\boldsymbol{\mathit{y}}}

\newcommand\xx{\boldsymbol{\mathit{x}}}

\newcommand{\Mc}{\mathcal{M}}
\newcommand{\Lc}{\mathcal{L}}                                                    
\newcommand{\rin}{\rho_{\text{in}}}
\newcommand{\rout}{\rho_{\text{out}}}
\newcommand{\Mg}{M_{\text{girl}}}
\newcommand{\Mb}{M_{\text{boy}}}
\newcommand{\Pin}{\Pi_{\text{in}}}
\newcommand{\Pout}{\Pi_{\text{out}}}
\newcommand{\st}{\text{s.t.}}

\title{Finding Stable Matchings\\ 
that are Robust to Errors in the Input\footnote{Supported in part by NSF grant CCF-1815901. Appeared in ESA 2018.}}

\author[1]{Tung Mai}
\author[1]{Vijay V.~Vazirani}

\affil[1]{University of California, Irvine}

\date{}

\begin{document}
	\maketitle

\begin{abstract}
We study the problem of finding solutions to the stable matching problem that
are robust to errors in the input and we obtain a polynomial time algorithm for a special class of errors.
In the process, we also initiate work on a new structural question concerning the stable matching 
problem, namely finding relationships between the lattices of solutions of two ``nearby'' 
instances.

Our main algorithmic result is the following:
We identify a polynomially large class of errors, $D$, that can be introduced in a 
stable matching instance.
Given an instance $A$ of stable matching, let $B$ be the random variable that represents the instance that results after introducing {\em one} error from $D$, chosen via a given discrete probability 
distribution. The problem is to find a stable matching for $A$ that maximizes the probability of
being stable for $B$ as well. Via new structural properties of the type described in the question
stated above, we give a combinatorial polynomial time algorithm for this problem.

We also show that the set of robust stable matchings for instance $A$, under probability distribution $p$, forms a sublattice of the lattice of stable matchings for $A$. We give an efficient algorithm for finding a succinct representation for this set; this representation has the property that any member of the set can be efficiently retrieved from it.
\end{abstract}

\section{Introduction}
\label{sec.intro}

Ever since its introduction in the seminal 1962 paper of Gale and Shapley \cite{GaleS}, 
the stable matching problem has been the subject of intense study from numerous different angles 
in many fields, including computer science, mathematics, operations research, economics and 
game theory, e.g., see the books \cite{Knuth-book, GusfieldI, Manlove-book}. The very first
matching-based market, namely matching medical interns to hospitals, was built around this problem,
e.g., see \cite{GusfieldI, Roth}.
Eventually, this led to an entire inter-disciplinary field, namely matching and market design
\cite{Roth}. The stable matching problem and market design were the
subject of the 2012 Nobel Prize in Economics, awarded to Roth and Shapley \cite{RS}. 

Another topic that has been extensively studied is the design of algorithms that produce robust solutions, e.g., see the books \cite{CaE:06, BEN:09}. 
Yet, there is a paucity of results at the intersection of these two topics. Indeed, we are aware of only two works \cite{aziz1,aziz2}; see Section \ref{sec.related} for a detailed description of these works. As explained there, their notion of robustness of a stable matching solution is quite distinct from ours. Our first contribution is a combinatorial polynomial time algorithm for finding a robust stable matching for a special case of allowable errors in the input. Our second contribution is to initiate work on a new structural question, namely finding relationships between the lattices of solutions of two ``nearby'' instances of stable matching. 

We note that a particularly impressive aspect of the stable matching problem is its deep and pristine
combinatorial structure. This in turn has led to efficient algorithms for numerous
questions studied about this problem, e.g., see the books mentioned above. In studying the new structural question, we have restricted ourselves to ``nearby'' instances which 
differ in only one agent's preference list. Clearly, this is only the tip of the iceberg as
far as ``nearby'' instances go. Moreover, the structural results are so clean and extensive -- see also our followup paper \cite{MV.Birkhoff} --
that they are likely to find algorithmic applications beyond the problem of finding robust 
solutions. In particular, with ever more interesting matching-based markets being designed and 
launched on the Internet \cite{Roth}, these new structural properties could find interesting 
applications and are worth studying further.

We will introduce our version of the problem of finding robust stable matchings via the following model:
Alice has an instance $A$ of the stable matching problem, over $n$ boys and $n$ girls, 
which she sends it to Bob over a channel that can introduce errors. Let $B$ denote the instance
received by Bob. Let $D$ denote a polynomial sized domain from which errors are introduced by 
the channel;
we will assume that the channel introduces at most one error from $D$. 
We are also given the discrete probability distribution, $p$ over $D$, from which the channel
picks {\em one} error. In addition,
Alice sends to Bob a matching, $M$, of her choice, that is stable for instance $A$.
Since $M$ consists of only $O(n)$ numbers of $O(\log n)$ bits each, as opposed to $A$ which requires 
$O(n^2)$ numbers, Alice is able to send it over an error-free channel. Now
Alice wants to pick $M$ in such a way that it is stable for instance $A$ and has the highest probability
of being stable for the instance $B$ received by Bob. Hence she picks $M$ from the set
$$ \arg\max_{N} \{ Pr_p[N \ \mbox{is stable for instance} \ B ~|~ N \ \mbox{is stable for instance} \ A] \}, $$
We will say that such a matching $M$ is {\em robust}.
We seek a polynomial time algorithm for finding such a matching.

Clearly, the domain of errors, $D$, will have to be well chosen to solve this problem.
A natural set of errors is 
{\em simple swaps}, under which the positions of two adjacent boys in a girl's list, or 
two adjacent girls in a boy's list, are interchanged.
We will consider a generalization of this class of errors, which we call {\em upward shift}.
For a girl $g$, assume her preference list in instance $A$ is 
$\{ \ldots,b_1,b_2, \ldots, b_k, b , \ldots \}$. 
Move up the position of $b$ so $g$'s list becomes $\{ \ldots, b, b_1,b_2, \ldots, b_k , \ldots \}$, 
and let $B$ denote the resulting instance. Then we will say that $B$ is obtained from $A$ by an 
upward shift. An analogous operation is defined on a boy $b$'s list. 
The domain $D$ consists of all such upward shifts on each boy's and each girl's list. 
Clearly, $|D|$ is $O(n^3)$, i.e., it is polynomially bounded. 
As will be clarified later, the operation of downward shift is much harder to deal with and we 
leave it as an open problem; see the Remark at the end of Section \ref{sec.lp}. In the rest of the paper, we will shorten ``upward shift'' to simply 
``shift''. Let us also clarify that we will deal with the generalization of stable matching in which incomplete preference lists are allowed. 

We next study the set of robust stable matchings for instance $A$ under probability distribution $p$. We show that this set forms a sublattice of the lattice of stable matchings for $A$ and we give an efficient algorithm for finding a succinct representation for this set. This representation has the property that any member of the set can be efficiently retrieved from it.

Our main theorem is the following.

\begin{theorem}
\label{thm.robust}
Given an instance $A$ of the stable matching problem, with possibly incomplete preference lists, and a probability distribution $p$ over the domain $D$ of errors defined above, there is an efficient algorithm that finds:
\begin{itemize}
	\item A robust stable matching for $A$.
	\item A partial order $\Pi_0$ on $O(n^2)$ elements such that its closed sets are isomorphic to the set of robust stable matching for $A$. The latter set forms a lattice $\Lc_0$ which is furthermore a sublattice of the lattice $\Lc_A$ of all stable matchings for $A$. 
\end{itemize}
The main computational step of our algorithm is to find one max-flow in a network with $O(n^2)$ vertices, where $n$ is the number of agents of each sex in instance $A$. 
\end{theorem}

\subsection{Overview of results and technical ideas}
\label{sec.overview}

Henceforth, we will assume that we are dealing with a stable matching instance with complete preference lists; we will finally generalize to incomplete lists in Section \ref{sec:incomplete}.

Let us first summarize some well-known structural facts, e.g., see \cite{GusfieldI}.
The set of stable matchings of an instance form a distributive lattice:
given two stable matchings $M$ and $M'$, their meet and join involve taking, for each boy,
the optimal or pessimal choice, respectively. It is easy to show that the resulting two matchings
are also stable. The extreme matchings of this lattice are called {\em boy optimal} and
{\em girl-optimal} matchings. A deep notion about this lattice is that of a rotation. A 
{\em rotation}, on an ordered list of $k$ boy-girl pairs, when applied to a matching $M$ in which 
all these boy-girl pairs are matched to each other, matches each boy to the next girl on the
list, closing the list under rotation. The $k$ pairs and the order among them are so chosen that 
the resulting matching is also stable; moreover, a rotation on a subset of these $k$ pairs, under
any ordering, leads to a matching that is not stable. 
Hence, a rotation can be viewed as a minimal
change to the current matching that results in a stable matching.
Rotations help traverse the lattice from the boy-optimal to the girl-optimal matching
along all possible paths available.

Birkhoff's \cite{Birkhoff} fundamental theorem for finite, distributive lattices shows that corresponding to each such lattice, $\Lc$, there is a partial order, say $\Pi$, such that the closed sets of $\Pi$ are isomorphic to the elements of $\Pi$. It turns out that for a lattice arising from a stable matching instance, the partial order $\Pi$ is defined on a set, say $R$, of rotations. Moreover, if $S$ is a closed set of $\Pi$, then starting in the lattice from the boy-optimal matching and applying the rotations in $S$ in any order consistent with a topological sort of $\Pi$, we will reach the stable matching corresponding to $S$.

Let $A$ and $B$ be two instances of stable matching over $n$ boys and $n$ girls, with sets of 
stable matchings $\Mc_A$ and $\Mc_B$, respectively, and lattices $\Lc_A$ and $\Lc_B$, respectively.
Then, it is easy to see that the matchings in $\Mc_A \cap \Mc_B$ form a sublattice in each of the
two lattices. Next assume that instance $B$ results from applying a shift operation, defined above,
to instance $A$. Then, we show that $\Mc_{AB} = \Mc_A \setminus \Mc_B$ is also a sublattice of $\Lc_A$ (Theorem \ref{cor:sublattice}). 
We use this fact crucially to show that there is at most one rotation, $\rin$, that leads from 
$\Mc_{A} \cap \Mc_{B}$ to $\Mc_{AB}$ and at most one rotation, $\rout$ that leads from 
$\Mc_{AB}$ to $\Mc_{A} \cap \Mc_{B}$ (Theorem \ref{thm:unique}).
Moreover, we can obtain efficiently this pair of rotations 
for each of the 
polynomially many instances that result from the polynomially many shifts (Proposition \ref{pro:computeRotation}).

Let $\Pi_A$ be the partial order for the lattice of instance $A$.
It is easy to see that a matching $M \in \Mc_A$, corresponding to a closed set $S$ of $\Pi_A$, 
is in $\Mc_B$ iff whenever $\rin \in S$, $\rout \in S$. 

In Section \ref{sec.lp} we give an integer program whose optimal solution
is a robust stable matching for the given probability distribution on shifts. The IP has
one indicator
variable, $y_{\rho}$, corresponding to each rotation $\rho$ in $\Pi_A$. The constraints of the program
ensure that the set $S$ of rotations that are set to 0 form a closed set. The rest of the 
constraints and the objective function ensure that the corresponding matching minimizes the
probability that $M$ is in $\Mc_A \setminus \Mc_B$.
We obtain the LP-relaxation of this IP and then obtain the dual LP. We interpret the latter as
solving a maximum circulation problem in a special network. This in turn is solvable as a max-flow problem on a network having $O(n^2)$ vertices and the solution yields an
integral optimal solution to the LP, hence yielding an efficient combinatorial  
algorithm for finding a robust stable matching.

In Section \ref{sec.rep}, we next study the set of robust stable matchings for instance $A$ under probability distribution $p$. We show that this set forms a sublattice of the lattice of stable matchings for $A$ (Lemma \ref{lem:sublattice}). We give an efficient algorithm for finding a succinct representation for this set. This representation has the property that any member of the set can be efficiently retrieved from it (Lemma \ref{lem:gen}).

\subsection{Related work}
\label{sec.related}

Aziz. et. al. \cite{aziz1} considered the problem of finding stable matching under uncertain linear preferences. They proposed three different uncertainty models:
\begin{enumerate}
\item Lottery Model: Each agent has a probability distribution
over strict preference lists, independent of other agents. 
\item Compact Indifference Model: Each agent has a single weak preference
list in which ties may exist. All linear order extensions of this
weak order have equal probability.
\item Joint Probability Model: A probability distribution over preference profiles
is specified.
\end{enumerate}
They showed that finding the matching with highest probability of begin stable is NP-hard for the Compact Indifference Model and the Joint Probability Model. For the very special case that preference lists of one gender are certain and the number of uncertain agents of the other gender are bounded by a constant, they gave a polynomial time algorithm that works for all three models.

The joint probability model is the most powerful and closest to our setting. The main difference is that in their model, there is no base instance, called $A$ in our model. The opportunity of finding new structural results arises from our model precisely because we need to consider two ``nearby'' instances, namely $A$ and the instance $B$ obtained by executing a shift.

Aziz. et. al. \cite{aziz2} introduced a pairwise probability model in which
each agent gives the probability of preferring one agent over another for all possible pairs.
They showed that the problem of finding a matching with highest probability of being stable is NP-hard even when no agent has a cycle in his/her certain preferences (i.e., the ones that hold with probability 1).

\subsubsection{A matter of nomenclature}
\label{sec.nomen}

Assigning correct nomenclature to a new issue under investigation is clearly critical for ease of 
comprehension. In this context we wish to mention that very recently,
Genc et. al. \cite{genc2} defined the notion of an $(a, b)$-supermatch as follows: this 
is a stable matching in which if any $a$ pairs break
up, then it is possible to match them all off by changing the partners
of at most $b$ other pairs, so the resulting matching is also stable.
They showed that it is NP-hard to decide if there is an
$(a, b)$-supermatch. They also gave a polynomial time algorithm for a very restricted version
of this problem, namely given a stable matching and a number $b$, decide if it is a 
$(1, b)$-supermatch. Observe that since the given instance may have exponentially many stable 
matchings, this does not yield a polynomial time algorithm even for deciding
if there is a stable matching which is a $(1, b)$-supermatch for a given $b$. 

Genc. et. al. \cite{genc1} also went on to defining the notion of the most robust stable matching, namely 
a $(1, b)$-supermatch where $b$ is minimum. We would like to point out that ``robust'' is a misnomer
in this situation and that the name ``fault-tolerant'' is more appropriate.
In the literature, the latter is used to
describe a system which continues to operate even in the event of failures and the former is 
used to describe a system which is able to cope with erroneous inputs, e.g., 
see the following pages from Wikipedia \cite{Robust, FT}.

	\section{Preliminaries}

\subsection{The stable matching problem}
The stable matching problem with incomplete preference lists takes as input a set of $n$ boys $B = \{b_1, b_2, \ldots , b_n\}$ and a set of $n$ girls $G = \{g_1, g_2, \ldots , g_n\}$ and a bipartite graph $G = (B, G, E)$, where $E$ is a set of edges connecting certain boys to certain girls. Each agent $v$ has a preference ranking the subset of opposite sex that are neighbors of $v$ in $G$. If $v$ cannot be paired to one of the neighbors, it prefers remaining unpaired rather than getting paired to any of the rest of the agents of opposite sex. The notation $b_i <_g b_j$ indicates that girl $g$ strictly prefers $b_j$ to $b_i$ in her preference list. Similarly, $g_i <_b g_j$ indicates that the boy $b$ strictly prefers $g_j$ to $g_i$ in his list.

A matching $M$ is a pairing of boys and girls so it is a maximum matching in $G$. Thus, if $G$ is the complete bipartite graph, $M$ will be a perfect matching in it. For each pair $bg \in M$, $b$ is called the partner of $g$ in $M$ (or $M$-partner), denoted by $p_M(g)$, and vice versa. For a matching $M$, a pair $bg \not \in M$ is said to be \emph{blocking} if $b$ and $g$ prefer each other to their partners in $M$; note that $b$ and/or $g$ need not be matched under $M$. Matching $M$ is \emph{stable} if there is no blocking pair w.r.t. $M$.

As stated in Section \ref{sec.overview}, from here on we will assume that $G$ is a complete bipartite graph. Finally, we will relax this assumption in Section \ref{sec.lp} to derive the more general result.

\subsection{The lattice of stable matchings}

Let $M$ and $M'$ be two stable matchings. We say that $M$ \emph{dominates} $M'$, denoted by 
$M \preceq M'$, if every boy weakly prefers his partner in $M$ to $M'$ (he either has the same partner or prefers his partner in $M$ to $M'$). It is well known that 
the dominance partial order over the set of stable matchings forms a 
distributive lattice \cite{GusfieldI}, with meet and join defined as follows.
The {\em meet} of $M$ and $M'$, $M \wedge M'$, 
is defined to be the matching that results when each boy chooses his more preferred partner 
from $M$ and $M'$; it is easy to show that this matching is also stable.
The {\em join} of $M$ and $M'$, $M \vee M'$, 
is defined to be the matching that results when each boy chooses his less preferred partner 
from $M$ and $M'$; this matching is also stable. These operations distribute, i.e.,
given three stable matchings $M, M', M''$,
$$ M \vee (M' \wedge M'') = (M \vee  M') \wedge (M \vee M'') \ \ \mbox{and} \ \
M \wedge (M' \vee M'') = (M \wedge M') \vee (M \wedge M'') .$$

It is easy to see that the lattice must contain a matching, $M_0$, that dominates all others
and a matching $M_z$ that is dominated by all others.
$M_0$ is called the \emph{boy-optimal matching}, since in it, each boy is matched to his most
favorite girl among all stable matchings. This is also the {\em girl-pessimal matching}.
Similarly, $M_z$ is the {\em boy-pessimal} or \emph{girl-optimal matching}.

\subsection{Rotations help traverse the lattice}
\label{sec.pre-rotations}

A crucial ingredient needed to understand the structure of stable matchings is the notion of 
a rotation, which was defined by Irving \cite{irving} and studied in detail in \cite{irving2}. 
A rotation takes $r$ matched pairs in a fixed order, say 
$\{b_0g_0, b_1g_1,\ldots, b_{r-1}g_{r-1}\}$ and ``cyclically'' changes the mates of these $2r$ 
agents, as defined below, to arrive at another stable matching. Furthermore, it represents a minimal
set of pairings with this property, i.e, if a cyclic change is applied on any subset of these 
$r$ pairs, with any ordering, then the resulting matching has a blocking pair and is not stable.
After rotation, the boys' mates weakly worsen and the girls' mates weakly improve. One can 
traverse from $M_0$ to $M_z$ by applying a suitable sequence of rotations, given by any topological sort of the rotation poset (defined below).

Let $M$ be a stable matching. For a boy $b$ let $s_M(b)$ denote the first girl $g$ on $b$'s list such that $g$ strictly prefers $b$ to her $M$-partner. Let $next_M(b)$ denote the partner in $M$ of girl $s_M(b)$. A \emph{rotation} $\rho$ \emph{exposed} in $M$ is an ordered list of pairs $\{b_0g_0, b_1g_1,\ldots, b_{r-1}g_{r-1}\}$ such that for each $i$, $0 \leq i \leq r-1$, $b_{i+1}$ is $next_M(b_i)$, where $i+1$ is taken modulo $r$. In this paper, we assume that the subscript is taken modulo $r$ whenever we mention a rotation. Notice that a rotation is cyclic and the sequence of pairs can be rotated. $M / \rho$ is defined to be a matching in which each boy not in a pair of $\rho$ stays matched to the same girl and each boy $b_i$ in $\rho$ is matched to $g_{i+1} = s_M(b_i)$. It can be proven that $M / \rho$ is also a stable matching. The transformation from $M$ to $M / \rho$ is called the \emph{elimination} of $\rho$ from $M$.

Let $\rho = \{b_0g_0, b_1g_1,\ldots, b_{r-1}g_{r-1}\}$ be a rotation. For $0 \leq i \leq r-1$, we say that $\rho$ \emph{moves $b_i$ from $g_i$ to $g_{i+1}$}, and \emph{moves $g_i$ from $b_{i}$ to $b_{i-1}$}. If $g$ is either $g_i$ or is strictly between $g_{i}$ and $g_{i+1}$ in $b_i$'s list, then we say that $\rho$ \emph{moves $b_i$ below $g$}. Similarly, $\rho$ \emph{moves $g_i$ above} $b$ if $b$ is $b_i$ or between $b_i$ and $b_{i-1}$ in $g_i$'s list.

\subsection{The rotation poset}

Let $A$ be an instance of stable matching problem and let $\Lc_A$ be the lattice of its stable matchings. Corresponding to lattice $\Lc_A$, there is a partial order $\Pi$ such that the closed subsets of $\Pi$ are in one-to-one correspondence with the matchings in $\Lc_A$; a \emph{closed subset} is a subset of the poset such that if an element is in the subset then all of its predecessors are also included. Since lattice $\Lc_A$ arises from a stable matching instance, it turns out that $\Pi$ is defined over a set of rotations. This partial order on rotations is called \emph{rotation poset}. Given a closed subset $C$, the corresponding matching $M$ is found by eliminating the rotations starting from $M_0$ according to the topological ordering of the elements in the subset. We say that $C$ \emph{generates} $M$. Let $C_1, C_2$ be closed subset generating $M_1$ and $M_2$. Then $C_1 \union C_2$, $C_1 \intersect C_2$ are closed subsets generating $M_1 \vee M_2$ and $M_1 \wedge M_2$ respectively.

\begin{lemma}[\cite{GusfieldI}, Lemma 3.2.1]
	\label{lem:pre2}
	For any boy $b$ and girl $g$, there is at most one rotation that moves $b$ to $g$, $b$ below $g$, or $g$ above $b$. Moreover, if $\rho_1$ moves $b$ to $g$ and $\rho_2$ moves $b$ from $g$ then $\rho_1 \prec \rho_2$.
\end{lemma}

\begin{lemma}[\cite{GusfieldI}, Lemma 3.3.2]
	\label{lem:computePoset}
	$\Pi$ contains $O(n^2)$ rotations and can be computed in polynomial time.
\end{lemma}

\begin{lemma}[\cite{GusfieldI}, Theorem 2.5.4]
	\label{lem:seqElimination}
	Every rotation appears exactly once in any sequence of elimination from $M_0$ to $M_z$.
\end{lemma}

	\section{Structural Results}
\label{sec:structure}

\subsection{The stable matchings in $\mathcal{M}_{AB}$ form a sublattice of $\Lc_A$}
 
Let $\mathcal{M}_{A}$ and $\mathcal{M}_{B}$ be the sets of all stable matchings under instance $A$ and $B$ respectively. Let $\mathcal{M}_{AB} = \mathcal{M}_{A} \setminus \mathcal{M}_{B}$. In other words, $\mathcal{M}_{AB}$ is the set of stable matchings in $A$ that become unstable in $B$. In this section we show that $\Mc_{AB}$ forms a lattice. 
We first prove a simple observation.

\begin{lemma} \label{lem:blockingPair}
	Let $M \in \Mc_{AB}$. The only blocking pair of $M$ under instance $B$ is $bg$.
\end{lemma}
\begin{proof}
	Since $M \not \in \mathcal{M}_B$, there must be a blocking pair $xy \not \in M$ under $B$. Assume $xy$ is not $bg$, we will show that $xy$ must also be a blocking pair in $A$. Let $y'$ be the partner of $x$ and $x'$ be the partner of $y$ in $M$. Since $xy$ is a blocking pair in $B$, $x >^B_y x'$ and $y >^B_x y'$. The preference list of $x$ remain unchanged from $A$ to $B$, so $y >^A_x y'$. Next, we consider two cases: 
	\begin{itemize}
		\item If $y$ is not $g$, the preference list of $y$ does not change. Therefore, $x >^A_y x'$, and hence, $xy$ is also a blocking pair in $A$.
		\item If $y$ is $g$, for all pairs $x,x'$ such that $x >^B_y x'$ and $x \not = b$, we also have $x >^A_y x'$. Therefore, $xy$ is a blocking pair in $A$.
	\end{itemize} 
	This contradicts the fact that $M$ is stable under $A$.
\end{proof}

Recall that $b_1 \geq_g b_2 \geq_g \ldots \geq_g b_k$ are $k$ boys right above $b$ in $g$'s list such that the position of $b$ is shifted up to be above $b_1$ in $B$. From Lemma~\ref{lem:blockingPair}, we can then characterize the set $\mathcal{M}_{AB}$.

\begin{lemma} \label{lem:characterize}
	$\mathcal{M}_{AB}$ is the set of all stable matchings in $A$ that match $g$ to a partner between $b_1$ and $b_k$ in $g$'s list, and match $b$ to a partner below $g$ in $b$'s list.
\end{lemma}
\begin{proof}
	Assume $M$ is a stable matching in $A$ that contains $b_i g$ for $1 \leq i \leq k$ and $bg'$ such that $g >_b g'$. 
	In $B$,
	$g$ prefers $b$ to $b_i$, and hence $bg$ is a blocking pair. Therefore, $M$ is not stable under $B$ and $M \in \mathcal{M}_{AB}$.
	
	To prove the other direction, let $M$ be a matching in $\mathcal{M}_{AB}$. By Lemma~\ref{lem:blockingPair}, $bg$ is the only blocking pair of $M$ in $B$. 
	For that to happen, $p_M(b) <^B_b g$ and $p_M(g) <^B_g b$. We will show that $p_M(g) = b_i$ for $1 \leq i \leq k$. Assume not, then $p_M(g) <^B_g b_k$, and hence, $p_M(g) <^A_g b$. Therefore, $bg$ is a blocking pair in $A$, which is a contradiction. 
\end{proof}

Let $\Lc_{A}$ be the boy-optimal lattice formed by $\Mc_{A}$. 

\begin{theorem} \label{cor:sublattice}
	The set $\mathcal{M}_{AB}$ forms a sublattice of $\mathcal{L}_A$.
\end{theorem}
\begin{proof}
	Assume $\mathcal{M}_{AB}$ is not empty. Let $M_1$ and $M_2$ be two matchings in $\mathcal{M}_{AB}$. By Lemma~\ref{lem:characterize}, $M_1$ and $M_2$ both match $g$ to a partner between $b_1$ and $b_k$ in $g$'s list, and match $b$ to a partner below $g$ in $b$'s list. Since $M_1 \wedge M_2$ is the matching resulting from having each boy choose the more preferred partner and each girl choose the least preferred partner, $M_1 \wedge M_2$ also belongs to the set characterized by Lemma~\ref{lem:characterize}. A similar argument can be applied to the case of $M_1 \vee M_2$.
	Therefore $\mathcal{M}_{AB}$ forms a sublattice of $\mathcal{L}_A$. 
\end{proof}

We will denote the lattice formed by $\mathcal{M}_{AB}$ as $\Lc_{AB}$.

	\subsection{Rotations going into and out of the sublattice $\Lc_{AB}$}

Let $M$ be a stable matching in $\Mc_A$ and $\rho$ be a rotation exposed in $M$ with respect to instance $A$. 
If $M \not \in \mathcal{S}$ and $M / \rho \in \mathcal{S}$ for a set $\mathcal{S}$ of stable matchings, we say that 
{\em $\rho$ goes into $\mathcal{S}$}.
Similarly, if $M \in \mathcal{S}$ and $M / \rho \not \in \mathcal{S}$, we say that 
{\em $\rho$ goes out of $S$}.
Let the set of all rotations going into $\mathcal{S}$ and out of $\mathcal{S}$ be $I_\mathcal{S}$ and $O_\mathcal{S}$, respectively. 

Let $\{b_{i_1}, \ldots b_{i_l}\}$ be the set of possible partners of $g$ in any stable matching in $\Mc_{AB}$, where $1 \leq i_1 \leq \ldots \leq i_l \leq k$. Let $\rho_1$ be a rotation moving $g$ to $b_{i_l}$, $\rho_2$ be the rotation moving $b$ below $g$ and $\rho_3$ be a rotation moving $g$ from $b_{i_1}$ (see \ref{sec.pre-rotations} for definitions). Note that each of $\rho_1, \rho_2$ and $\rho_3$ might not exist.

\begin{lemma} \label{lem:rotation}
	$I_{\Mc_{AB}}$ can only contain $\rho_1$, $\rho_2$. $O_{\Mc_{AB}}$ can only contain $\rho_3$.
\end{lemma}

\begin{proof}	 
	Consider a rotation $\rho \in I_{\Mc_{AB}}$. There exists $M \in \Mc_{A} \setminus \Mc_{AB}$ such that $M / \rho \in \Mc_{AB}$. 
	By Lemma~\ref{lem:characterize}, $M / \rho$ matches $g$ to a partner between $b_1$ and $b_k$ in $g$'s list, and matches $b$ to a partner below $g$ in $b$'s list.
	Moreover, $M$ either does not contain $b_ig$ for any $1 \leq i \leq k$, or contains $bg'$ where $g' \geq_{b} g$, or both.
	If $M$ does not contain $b_ig$ for any $1 \leq i \leq k$, then $\rho = \rho_1$. If $M$ contains $bg'$ where $g' \geq_{b} g$, then $\rho = \rho_2$.
	
	Consider a rotation $\rho \in O_{\Mc_{AB}}$. There exists $M \in \Mc_{AB}$ such that $M / \rho \in \Mc_{A} \setminus \Mc_{AB}$. 
	Again, by Lemma~\ref{lem:characterize}, $M$ contains $b_ig$ for $1\leq i \leq k$ and $bg'$ where $g' <_{b} g$. Since $M$ dominates $M / \rho$ in the boy optimal lattice, $b$ must prefer $g'$ to his partner in $M / \rho$.
	Hence, $M / \rho$ matches $b$  to a partner below $g$ in $b$'s list. 
	Therefore, $M / \rho$ must not contain $b_ig$ for any $1 \leq i \leq k$.
	It follows that $\rho$ must be $\rho_3$.
\end{proof}

\begin{lemma} \label{lem:rotationComparision}
	If both $\rho_1$ and $\rho_2$ exist then $\rho_1 \preceq \rho_2$.
\end{lemma}

\begin{proof}
Assume that $\rho_1 \not = \rho_2$ and there exists a sequence of rotation eliminations, from $M_0$ to a stable matching $M$ in which $\rho_2$ is exposed, that does not contain $\rho_1$.
Since $\rho_2$ moves $b$ below $g$,  $g$ is matched a partner higher than $b$ in her list in $M / \rho_2 $. Therefore, the partner can only be $b_{i_l}$ or a boy higher than $b_{i_l}$ in $g$'s list.

Consider any sequence of rotation eliminations from $M / \rho$ to $M_z$. 
In the sequence, the position of $g$'s partner can only go higher in her list.
Therefore, $\rho_1$ cannot be exposed in any matching in the sequence.
It follows that $\rho_1$ is not exposed in a sequence of eliminations from $M_0$ to $M_z$, which is a contradiction by Lemma~\ref{lem:seqElimination}.
\end{proof}

\begin{theorem} \label{thm:unique}
	There is at most one rotation in $I_{\Mc_{AB}}$ and at most one rotation in $O_{\Mc_{AB}}$. Moreover, the rotation in $I_{\Mc_{AB}}$ must be either $\rho_1$
	or $\rho_2$, and the rotation in $O_{\Mc_{AB}}$ must be $\rho_3$.
\end{theorem}
 
\begin{proof}
	By Lemma~\ref{lem:rotation},  $I_{\Mc_{AB}}$ can contain at most 2 rotations, namely $\rho_1$ and $\rho_2$ if they are distinct. By Lemma~\ref{lem:rotationComparision}, if both of them exist, $\rho_1 \preceq \rho_2$. Hence, $I_{\Mc_{AB}}$ can contain at most one rotation, and it is either $\rho_1$ or $\rho_2$.
	
	Again, by Lemma~\ref{lem:rotation}, $O_{\Mc_{AB}}$ can contain at most one rotation, namely $\rho_3$ if it exists.
\end{proof}

By Theorem~\ref{thm:unique}, there is at most one rotation $\rin$ coming into $\Mc_{AB}$ and at most one rotation $\rout$ coming out of $\Mc_{AB}$. 

\begin{proposition} \label{pro:computeRotation}
	$\rin$ and $\rout$ can be computed in polynomial time. 
\end{proposition}
\begin{proof}
	Since we can compute $\Pi_A$ efficiently according to Lemma~\ref{lem:computePoset}, each of $\rho_1$, $\rho_2$ and $\rho_3$ can be computed efficiently.
	
	First we can check possible partners of $b$ and $g$ with respect to instance $A$.
	By Lemma~\ref{lem:characterize}, $\Mc_{AB}$ is empty if none of the possible partners of $g$ is between $b_1$ and $b_k$ in $g$'s list or none of the partners of $b$ is below $g$ in $b$'s list. 
	It follows that both $\rin$ and $\rout$ do not exist.
	Hence we may assume that such a case does not happen. 
	
	Suppose $\rho_2$ exists. If $\rho_3$ exists and $\rho_3 \preceq \rho_2$, $\Mc_{AB} = \emptyset$. Otherwise, $\rin = \rho_2$, and $\rout = \rho_3$ if $\rho_3$ exists.  
	
	Suppose $\rho_2$ does not exist. If $\rho_1$ exists, $\rin = \rho_1$. If $\rho_3$ exists, $\rout = \rho_3$.
\end{proof}

\begin{lemma} \label{lem:subset}
	Let $M$ be a matching in $\Mc_{AB}$ and $S$ be the corresponding closed subset in $\Pi_{A}$. If $\rho_1$ exists, $S$ must contain $\rho_1$. If $\rho_2$ exists, $S$ must contain $\rho_2$. If $\rho_3$ exists, $S$ must not contain $\rho_3$.
\end{lemma}	

\begin{proof}
	If $\rho_1$ exists, $M_0$ does not contain $b_ig$ for any $i \in [1,k]$. Since $M \in \Mc_{AB}$, by Lemma~\ref{lem:characterize} $M$ matches $g$ to a boy between $b_1$ and $b_k$ in her list. The set of rotations eliminated from $M_0$ to $M$ must include $\rho_1$.
	
	If $\rho_2$ exists, $b$ can not be below $g$ in $M_0$. Since $b$ is below $g$ in $M$, by Lemma~\ref{lem:characterize} the set of rotations eliminated from $M_0$ to $M$ must include $\rho_2$.
	
	Assume that $\rho_3$ exists and $S$ contains $\rho_3$. Since $\rho_3$ moves $g$ up from $b_{i_1}$, $M$ can not contain $b_ig$ for any $i \in [1,k]$. This is a contradiction.
\end{proof}

\subsection{The rotation poset for the sublattice $\Lc_{AB}$}

From the previous section we know that $\Lc_{AB}$ is a sublattice of $\Lc_A$. 
In this section we give the rotation poset that generates all stable matchings in this sublattice.

We may assume that $M_{AB} \not = \emptyset$. If $\rin$ exists, let $\Pi_{\text{in}} = \{ \rho \in \Pi_A: \rho \preceq \rin \}$ and $\Mb$ be the matching generated by $\Pi_{\text{in}}$. Otherwise, let $\Mb = M_0$. Similarly, let $\Mg$ be the matching generated by $\Pi_A \setminus \Pi_{\text{out}}$, where $\Pi_{\text{out}} = \{ \rho \in \Pi_A: \rho \succeq \rout \}$, if $\rout$ exists, and $\Mg = M_z$ otherwise. 

\begin{lemma}
	$\Mb$ is the boy-optimal matching in $\Mc_{AB}$, and $\Mg$ is the girl-optimal matching in $\Mc_{AB}$.
\end{lemma}

\begin{proof}
	Let $M$ be a matching in $\Mc_{AB}$ generated by a closed subset $S \subseteq \Pi_A$. By Lemma~\ref{lem:subset}, if $\rin$ exists, $S$ must contain $\rin$. Since $\Pi_{\text{in}}$ is the minimum set containing $\rin$, $\Pi_{\text{in}} \subseteq S$. Therefore, $\Mb \preceq M$.
	
	To prove that $M \preceq \Mg$, we show $S \subseteq \Pi_A \setminus \Pi_{\text{out}}$. Assume otherwise, then there exists a rotation $\rho \in S$ such that $\rho \not \in \Pi_A \setminus \Pi_{\text{out}}$. It follows that $\rho \in \Pi_{\text{out}}$, and hence $\rho \succeq \rout$. Since $S$ contains $\rho$ and $S$ is a closed subset, $S$ must also contain $\rout$. This is a contradiction by Lemma~\ref{lem:subset}.
\end{proof}

\begin{theorem}
	\label{lem:rhoAB}
	$\Pi_{AB} = \Pi_A \setminus \left( \Pin \cup \Pout \right) $ is the rotation poset generating $\Lc_{AB}$.
\end{theorem}

\begin{proof}
	Let $M$ be a matching in $\Mc_{AB}$ generated by a closed subset $S \subseteq \Pi_A$. Let $S' = S \setminus \Pi_{\text{in}}$.
	We show that $S'$ is a closed subset of $\Pi_{AB}$ and
	eliminating the rotations in $S'$ starting from $\Mb$ according to the topological ordering of the elements gives $M$.
	 
	First $S' \cap \Pi_{\text{in}} = \emptyset$ trivially. Since $M \in \Mc_{AB}$, $S$ does not contain $\rout$ by Lemma~\ref{lem:subset}. Therefore, $S'$ does not contain $\rout$, and $S' \cap \Pi_{\text{out}} = \emptyset$. It follows that $S'$ is a closed subset of $\Pi_{AB}$.
	
	Next observe that we can eliminate rotations in $S$ from $M_0$ by eliminating rotations in $\Pin$ first and then eliminating rotations in $S \setminus \Pin$. This can be done because $\Pin$ is a closed subset of $\Pi_A$. Since $\Pin$ generates $M$, the lemma follows.
\end{proof}

Finally, we observe that the results stated above also follow when we make an upward shift in a boy's list.

\begin{lemma}
	 Let $A$ be an instance of stable matching, and $B$ be another instance obtained by introducing a shift in the list of a boy in instance $A$. Then there is at most one rotation, 
	 $\rin$, that leads from 
	 $\Mc_{A} \cap \Mc_{B}$ to $\Mc_{AB}$ and at most one rotation, 
	 $\rout$ that leads from 
	 $\Mc_{AB}$ to $\Mc_{A} \cap \Mc_{B}$.
\end{lemma}
\begin{proof}
Let us switch the roles of boys and girls, and reverse all partial orders in $\Mc_A$ and $\Pi_{A}$. Let $\overline{\Mc}_A$ and $\overline{\Pi}_{A}$ be the resulting matching lattice and rotation poset. Let $\overline{\rho}_{\text{in}}, \ \overline{\rho}_{\text{out}} \in \overline{\Pi}_{A}$ be the rotations leading from $\Mc_{A} \cap \Mc_{B}$ to $\Mc_{AB}$ and from
$\Mc_{AB}$ to $\Mc_{A} \cap \Mc_{B}$ in the this lattice, as guaranteed by Theorem~\ref{thm:unique}.
Then $\rho_{\text{in}} = \overline{\rho}_{\text{out}}$ and $\rho_{\text{out}} = \overline{\rho}_{\text{in}}$.
\end{proof}

	\section{Efficient Algorithm for Finding a Robust Stable Matching }
\label{sec.lp}

As stated in Section~\ref{sec.intro}, let $D$ be the domain of all possible shifts applied to each boy's list and each girl's list in instance $A$ and let $p$ be a discrete probability distribution on ${D}$. Pick {\em one} shift from $D$ under $p$ and let $C$ be the random variable denoting the resulting instance. As defined in Section~\ref{sec.intro}, a robust stable matching is a stable matching $M \in \Mc_{A}$ that minimizes the probability that $M \in \Mc_{AC}$.

For a particular choice of shift from $D$, let $B$ denote the resulting stable matching instance and  let $\rin^B$ and $\rout^B$ denote the rotations going into $\Mc_{AB}$ and out of $\Mc_{AB}$, respectively. 
By Proposition~\ref{pro:computeRotation}, $\rin^B$ and $\rout^B$ can be computed efficiently for each such $B$. For convenience, we will name the chosen shift also as $B$.

By Lemma~\ref{lem:computePoset}, $\Pi_A$ can be computed in polynomial time. We add two additional vertices to $\Pi_A$, a source $s$ preceding all other vertices and a sink $t$ succeeding all other vertices. 
For a shift $B$, we may ignore the cases where neither $\rin^B$ nor $\rout^B$ exist. In that case, either $\Mc_{A} = \Mc_{B}$ or $\Mc_{A} \intersect \Mc_{B} = \emptyset$. 
Hence, we may assume that $\Mc_{AB}$ is always a proper non-empty subset of $\Mc_{A}$.
For a shift $B$ such that $\rin^B$ does not exist, let $\rin^B = s$. Similarly, for a shift $B$ such that $\rout^B$ does not exist, let $\rout^B = t$.

Let $H$ be the Hasse diagram of $\Pi_A \union \{s,t\}$, defined as follows: The Hasse diagram of a poset is a directed graph with a vertex for each element in poset, and an
edge from $x$ to $y$ if $x \prec y$ and there is no $z$ such that $x \prec z \prec y$. In other words, all precedence relations implied by transitivity are suppressed.

For each $B$, let $e_B = (\rho_{\text{out}}^B , \rho_{\text{in}}^B )$ and $F = \{ e_B | B \in D\}$. The integer program is as follows: 

\begin{equation}
\label{ip}
\tag{IP}
\begin{aligned}
\min  &~~~ \sum_{B} x_B p_B  \\
\st   &~~~ y_{u} - y_{v} \leq 0 &\quad & \forall (u,v) \in E(H) \\
&~~~ y_t = 1, ~ y_s = 0 \\
&~~~ x_B \geq y_{u} - y_{v}  &\quad & \forall (u,v) = e_B \in F \\
&~~~ x_B \geq 0  &\quad & \forall B \\
&~~~ y_v \in \{0,1\} &\quad &\forall v \in H. 
\end{aligned}
\end{equation}

\begin{lemma}
	\label{lem.correctness}
	An optimal solution to (\ref{ip}) gives a robust stable matching.
\end{lemma}
\begin{proof}
	Let $S = \{\rho: y_v = 0\}$. The set of constraints:
	\[ y_{u} \leq y_{v} \quad  \forall (u,v) \in E(H) \] 
	guarantees that $S$ is a closed subset. 
	
	Notice that $x_B = 1$ if and only if $y_{\rout^B} = 1$ and $y_{\rin^B} = 0$. This, in turn, happens if and only if the matching generated by $S$ is in $\Mc_{AB}$.
	Therefore, by minimizing $\sum_{B} x_B p_B$, we can find a closed subset that generates a robust stable matching. 
\end{proof}


Next, consider the LP-relaxation of this IP:

\begin{equation}
\label{lp}
\tag{LP}
\begin{aligned}
\min  &~~~ \sum_{B} x_B p_B  \\
\st   &~~~ y_{u} - y_{v} \leq 0 &\quad & \forall (u,v) \in E(H) \\
&~~~ y_t - y_s = 1 \\
&~~~ x_B \geq y_u - y_v  &\quad & \forall e_B = (u,v) \in F \\
&~~~ x_B \geq 0  &\quad & \forall B \\
\end{aligned}
\end{equation}

Define $f_{uv}$, $f_{st}$ and $g_{uv}$ to be the dual variables corresponding to the first three constraints of (\ref{lp}). Then its dual is:
 
\begin{equation}
\label{dp}
\tag{DP}
\begin{aligned}
\max  &~~~ f_{st} \\
\st   &~~~ \sum_{v: (v,u) \in E(H)} f_{vu}  + \sum_{v: (v,u) \in F} g_{vu} = \sum_{v: (v,u) \in E(H)} f_{uv}  + \sum_{v: (v,u) \in F} g_{uv} &\quad& \forall u \in H\\
&~~~ g_{uv} \leq p_B &\quad& \forall e_B = (u,v) \in F \\
&~~~ f_{uv} \geq 0 &\quad& \forall  (u,v) \in E(H) \\
&~~~ g_{uv} \geq 0 &\quad& \forall  (u,v) \in F
\end{aligned}
\end{equation}

We will interpret (\ref{dp}) as solving a maximum circulation problem in the following network $N$: It has three sets of edges, $\{(s, t)\}, E(H)$ and $F$. The edges in $E(H)$ are of infinite capacity and the flow on $e \in E(H)$ is denoted by $\ff_e$. $F$ contains edges with capacity $\pp$, and the flow on $e \in F$ is denoted by $\gg_e$. The goal is to push the maximum amount of flow from $t$ to $s$ through $E(H) \union F$ and then back from $s$ to $t$ on edge $(s, t)$ of infinite capacity. 

\begin{lemma}
	\label{lem:combinatorial}
	(\ref{lp}) always has an integral optimal solution and there is a combinatorial polynomial time algorithm for solving it.
\end{lemma}

\begin{proof}
First, remove edge $st$ from the above described network and find a maximum flow from $t$ to $s$ and denote it by $\ff, \ \gg$. Obtain the residual graph, which clearly will not have any paths from $t$ to $s$. Let $R$ be the set of vertices reachable from $t$ using residual edges. Construct $(\xx,\yy)$ as follows: 
	 \[
	y_u = \begin{cases}
	1 & \text{if }  u \in R \\
	0 & \text{otherwise } 
	\end{cases}\\ 
	\]
	\[
	x_B = \min\{ y_{\rho_{\text{out}}^B} - y_{\rho_{\text{in}}^B}, 0\}
	\]
	Clearly, $(\xx,\yy)$ is integral. Moreover, $(\xx,\yy)$ and $(\ff,\gg)$ satisfy complementarity: 
	\begin{itemize}
		\item $x_B(g_{uv} - p_B) = 0$ for all $e_B = (u, v) \in F$.
		\item $f_{uv}(y_u - y_v) = 0$ for all $(u,v) \in E(H)$.
		\item $g_{uv}(y_u - y_v - x_B) = 0$ for all $e_B = (u, v) \in F$.
	\end{itemize}
	Hence, $(\xx,\yy)$ is an integral optimal solution for (\ref{lp}).
\end{proof}

{\bf Remark.}  Notice that we formulated the IP for robust stable matching as a minimization problem, i.e, minimizing the probability that the matching $M$ is in $\Mc_{AB}$. The reason is that this involves testing only whether $S$ crosses the edge $(\rout^B, \rin^B)$. On the other hand, checking whether $M$ is in $\Mc_A \cap \Mc_B$ involves testing two edges in general, i.e,
that $S$ crosses $(t, \rout^B)$ or $(\rin^B, s)$. The latter would not have led to a linear IP.

{\bf Remark.} 
We can now explain why downward shifts are much more difficult to deal with.
Observe that the guarantee that there is at most one rotation, $\rin$, that leads from 
$\Mc_{A} \cap \Mc_{B}$ to $\Mc_{AB}$ and at most one rotation, $\rout$ that leads from 
$\Mc_{AB}$ to $\Mc_{A} \cap \Mc_{B}$ is crucial for forming the IP and showing that its LP-relaxation always has integral an integral solution. For the case of downward shifts, there may be more than one rotation that lead from $\Mc_{A} \cap \Mc_{B}$ to $\Mc_{AB}$ and from 
$\Mc_{AB}$ to $\Mc_{A} \cap \Mc_{B}$. Although we can still formulate an IP for this case, its LP-relaxation is not guaranteed to have an integral solution.

\subsection{Extending to incomplete preference lists}
\label{sec:incomplete}

Finally, we show how to extend our algorithm to the generalization of stable matching to incomplete preference lists. It is well known that the set of unmatched agents remain unchanged
under all stable matchings in this case, e.g., see \cite{GusfieldI}. Let $A$ be the given instance and let $B$ be obtained by executing an upward shift on one agent's list. If the set of unmatched agents under $B$ is not the same as under $A$, then $\Mc_{A} \intersect \Mc_{B} = \emptyset$ and we can ignore $B$. Otherwise, the sets of stable matchings in $A$ and $B$ are defined over the same subset of the agents and all our results carry over. 

The first part of
Theorem~\ref{thm.robust} now follows from Lemmas \ref{lem.correctness} and \ref{lem:combinatorial}.

\section{Succinct Representation for the Sublattice of Robust Stable Matchings}
\label{sec.rep}

We first prove that the set of robust stable matchings forms a sublattice of the lattice of all stable matchings of the given instance. We then use our combinatorial solution of (\ref{dp}) to show how to obtain a succinct structure that helps generate all matchings from this sublattice.
 
\begin{lemma}
	\label{lem:sublattice}
	The set of robust stable matchings to instance $A$ under probability distribution $p$ forms a sublattice of $\Lc_A$. 
\end{lemma}
\begin{proof}
	Let $M_1$ and $M_2$ be two robust stable matchings and let $S_1$ and $S_2$ be the corresponding closed subsets. It suffices to show that the matchings generated by $S_1 \union S_2$ and $S_2 \intersect S_2$ are also robust. We say that a closed subset $S$ \emph{separates} and edge  $(\rho_{\text{out}}^B, \rho_{\text{in}}^B )$ in $E$ if $\rho_{\text{out}}^B \not \in S$ and $\rho_{\text{in}}^B \in S$.
	
	Divide the edges in $F$ into 5 sets as follows: 
	\begin{enumerate}
		\item $F_1$ is the set of edges in $F$ from $V \setminus (S_1 \union S_2)$ to $S_1 \intersect S_2$.
		\item $F_2$ is the set of edges in $F$ from $S_1 \setminus S_2$ to $S_1 \intersect S_2$.
		\item $F_3$ is the set of edges in $F$ from $S_2 \setminus S_1$ to $S_1 \intersect S_2$.
		\item $F_4$ is the set of edges in $F$ from $V \setminus (S_1 \union S_2)$ to $S_1 \setminus S_2$.
		\item $F_5$ is the set of edges in $F$ from $V \setminus (S_1 \union S_2)$ to $S_2 \setminus S_1$.
	\end{enumerate}

	Let $P_i = \sum_{e_B \in F_i} p_B$ for each $1 \leq i \leq 5$.
	Since $M_1$ and $M_2$ are robust stable matchings, the objectives obtained by $S_1$ and $S_2$ in (\ref{ip}) are equal:
	\[ \sum_{S_1 \text{ separates } e_B} p_B =  \sum_{S_2 \text{ separates } e_B} p_B\]
	Therefore,
	\begin{align*}
	P_1 + P_2 + P_5 &= P_1 + P_3 + P_4 \\
	P_2 + P_5 &= P_3 + P_4.
	\end{align*} 
	
	We will show that $P_2 = P_4$ and $P_3 = P_5$. Assume without loss of generality that $P_2 > P_4$ and $P_3 > P_5$. Then 
	\[P_1 + P_2 + P_3 > P_1 + P_2 + P_5. \]
	In other words, the objective of (\ref{ip}) obtained by $S_1 \intersect S_2$ is smaller than the one obtained by $S_1$. This contradicts the fact that $M_1$ is robust stable matching.
	Therefore, $P_2 = P_4$ and $P_3 = P_5$ as desired.
	
	It follows that $P_1 + P_2 + P_3$ and $P_1 + P_4 + P_5$ also attain the minimum value of the objective function of (\ref{ip}). Hence, the matchings generated by $S_1 \union S_2$ and $S_2 \intersect S_2$ are also robust.
\end{proof}

By Lemma~\ref{lem:sublattice}, the set of robust stable matchings is a sublattice, say $\Lc_0$, of $\Lc_A$. By Birkhoff's Theorem \cite{Birkhoff} we know that there is a partial order, say $\Pi_0$, whose closed sets are isomorphic to $\Lc_0$. 
Next, we show how to construct $\Pi_0$ using our combinatorial solution for (\ref{dp}).

As in Lemma \ref{lem:combinatorial}, remove edge $st$ from network $N$ described above and find a maximum flow from $t$ to $s$ and denote it by $\ff, \ \gg$. Obtain the residual graph, say $G$.
The strongly connected components of $G$ are the elements of $\Pi_0$. 
Contracting the strongly connected components of $G$ yields a DAG $D$, which gives the precedence relations in $\Pi_0$. 
To be precise, $x \prec y$ in $\Pi_0$ if and only if there is a path from $x$ to $y$ in $D$.  

\begin{lemma}
	The closed sets of partial order $\Pi_0$ correspond exactly to robust stable matchings.
\end{lemma}

\begin{proof}
	Let $S$ be a closed set in $\Pi_0$. There are no edges in the residual graph that go from $\overline{S}$ to $S$. Hence all edges in the cut $(\overline{S}, S)$ are fully saturated and therefore this cut minimizes the objective function of (\ref{ip}). Hence the corresponding matching is a robust stable matching. The reverse direction is straightforward.
\end{proof}

Recall that if $\Pi$ is the rotation poset for lattice $\Lc_A$, then the matching corresponding to any closed set $S$ in $\Pi$ is obtained by starting from the boy-optimal matching in $\Lc_A$ and applying all rotations in $S$ in any order consistent with a topological sort of $\Pi$. In \cite{MV.Birkhoff} we show that the corresponding process for finding the matching in $\Lc_0$ corresponding to a closed set $S$ of $\Pi_0$ is the following: The elements of $S$ are sets of rotations. Let $U$ be the union of all these sets. Now starting from the boy-optimal matching in $\Lc_A$, apply all rotations in $U$ in any order consistent with a topological sort of $\Pi$. This yields the matching in $\Lc_0$ corresponding to set $S$. Hence we get.

\begin{lemma}
\label{lem:gen}
	$\Pi_0$ generates the sublattice $\Lc_0$ of robust stable matchings.
\end{lemma}

The second part of Theorem \ref{thm.robust} follows from Lemmas \ref{lem:sublattice} and \ref{lem:gen}.

	\section{Discussion}
\label{sec.discussion}

As stated in the Introduction, the two main questions on stable matching considered in this paper
are obtaining efficient algorithms for finding solutions that are robust to errors in
the input, and the structural question of finding 
relationships between the lattices of solutions of two ``nearby'' instances. The current paper
and our followup work \cite{MV.Birkhoff} seem to suggest that both these issues are likely
to lead to much work in the future. In particular, the structural results are so clean and extensive
that they are likely to find algorithmic applications beyond the problem of finding robust 
solutions. One possible domain of applications that may be able to exploit these structural
properties is matching-based markets, particularly as we are seeing ever more interesting such 
markets being designed and launched on the Internet, e.g., see \cite{Roth}.

At a more detailed level, we have left the open problem of dealing with downward shifts. The
domain $D$, for which we have obtained our algorithm, is very restrictive and we need to 
extend it to a larger domain. Our followup paper \cite{MV.Birkhoff} partially does this,
though it works with a weaker notion of ``robust''. It seems more should be doable; in particular, what happens if two or more errors are introduced simultaneously? 

Beyond these questions, pertaining to the most basic of formulations of stable matching, one can 
study numerous variants and generalizations, e.g., the stable
roommates problem, and matching intern couples to hospitals. Each of these bring
their own structural properties and challenges, e.g., see \cite{GusfieldI, Manlove-book}.

	\section{Acknowledgements}
\label{sec.ack}

We wish to thank David Eppstein, Mike Goodrich and Sandy Irani for interesting discussions
that sparked off the question addressed in this paper.

	\bibliographystyle{alpha}
	\bibliography{refs}
\end{document}